\documentclass[11pt]{amsart}

\usepackage[percent]{overpic}
\usepackage{graphicx,subfigure,amsmath,amssymb,amsfonts,bm,epsfig,epsf,url,dsfont,tcolorbox,bbm,microtype}
\usepackage{amsthm,mathrsfs}
\usepackage{tikz}
\usepackage{multirow}
\usepackage{times}
\usepackage{booktabs}
\usepackage{cases}
\usepackage{enumitem}
\usepackage[font=footnotesize]{subcaption}
\makeatletter
\renewcommand{\@captionfont}{\footnotesize}
\makeatother

\usepackage{tikz}
\usepackage{marginnote}

\usepackage{fancybox}

\usepackage{xcolor}

\usepackage[margin=1.15in]{geometry}
\usepackage{hyperref}
\hypersetup{
    colorlinks,
    linkcolor={blue!80!black},
    citecolor={red!80!black},
    urlcolor={blue!80!black}
}

\newcommand{\graybox}[1]{%
\begin{tcolorbox}[
  width=\linewidth,
  left=1mm,
  right=1mm,
  boxsep=0pt,
  boxrule=0.5pt,
  sharp corners,
  colback=white!95!black]
#1
\end{tcolorbox}
}

\usetikzlibrary{shapes.geometric}

\theoremstyle{plain}
\newtheorem{theorem}{Theorem}[section]
\newtheorem{lemma}[theorem]{Lemma}

\newtheorem{corollary}[theorem]{Corollary}
\theoremstyle{remark}
\newtheorem{remark}[theorem]{Remark}

\newcommand{\R}{\mathbb{R}}
\newcommand{\C}{\mathbb{C}}

\newcommand{\E}{\mathbb{E}}
\newcommand{\la}{\langle}
\newcommand{\ra}{\rangle}
\newcommand{\tr}{\operatorname{tr}}
\newcommand{\Stab}{\operatorname{Stab}}
\newcommand{\cH}{\mathcal H}
\newcommand{\cR}{\mathcal R}
\newcommand{\Kt}{\tilde K}
\newcommand{\lmax}{\lambda_{\max}}
\newcommand{\wt}{\operatorname{wt}}

\makeatletter
\renewcommand{\@makefnmark}{\hbox{\@textsuperscript{\normalfont\@thefnmark}}}
\makeatother

\usepackage{etoolbox}

\makeatletter
\AtEndEnvironment{remark}{\hfill$\triangleleft$}
\makeatother
\title[Improvements of the MRRW bounds]{Comments on the recent improvements of the MRRW bounds}
\author{Alexander Barg}\address{University of Maryland}\email{abarg@umd.edu}
\date{}

\begin{document}

\begin{abstract} The asymptotic McEliece–Rodemich–Rumsey–Welch bound (1977) limits the largest attainable rate of binary codes as a function of the relative distance. After a nearly half-century hiatus, this result was recently improved in two concurrent works, by OpenAI and by O. Alrabiah and V. Guruswami. The two arguments look entirely different, a Delsarte certificate on the one hand, a classical–quantum channel and the pretty good measurement on the other, and they yield the same bound. The purpose of this note is to explain why: in both proofs, a subspace is attached to every codeword and moved with it, and the bound counts how many such subspaces fit in the ambient space, exactly in the first case and in the probabilistic sense of typicality in the second. We also present the OpenAI proof in the language and context of coding theory, as an extension of the spectral method in which the single vector attached to a codeword is replaced by a subspace.\end{abstract}

\maketitle

{\small{\tableofcontents}}

\section{Introduction}
The paper of McEliece-Rodemich-Rumsey-Welch \cite{MRRW77}, contains two related results on the asymptotic rate of binary codes and of binary constant-weight codes\footnote{This note is written as a set of comments, and we assume that the readers are familiar with the general problem of asymptotic upper bounds on the rate of codes.}. The proofs are based on the same idea, applied in the Hamming association scheme and the Johnson scheme. The constant-weight bound can be translated to the Hamming space either by the Bassalygo-Elias inequality or by its functional analog known as Rodemich's theorem. The resulting bounds on codes in the Hamming space are usually called MRRW-1 and MRRW-2. The second of them is strictly better than the first for relative distances $0<\delta< 0.273$ (the upper limit is approximate) and coincides with the first for larger $\delta$. These bounds have been the best known results for nearly a half-century since their appearance. In the meantime, they were proved in a number of different ways, but all attempts to surpass them have been unsuccessful --- until recently (August 2026), when two concurrent works \cite{OAI26}, \cite{AG26} presented a small but conceptually important improvement of these estimates. The manuscript \cite{OAI26} in addition contains improvements of the Kabatiansky-Levenshtein upper bound on the asymptotic rate of spherical codes and the exponent of the packing density of $\R^n$, and the second work \cite{AG26} proves a new upper bound on the rate of low-density parity-check codes. Both works also contain extensive discussions of the background and related results. 

In this note we focus on the MRRW-1 bound. While it is not the best result numerically, the concepts behind the improvements are very similar, and the Johnson case is technically more involved, so the MRRW-1 case suffices to highlight the new ideas behind the proof. What is more, the improvement of the Kabatiansky--Levenshtein bound on the rate of spherical codes in \cite{OAI26} also relies on similar ideas.

This note claims no novelty beyond the presentation and a few contextual remarks. 
We present the proof of the bound in \cite{OAI26} building up from the one-dimensional case of \cite{BN06}: as it turns out, once the details there are explained, the generalization is easy to process. Our perspective also reveals connections between the arguments underlying the Delsarte-type proof of \cite{OAI26} and the decoding-based proof for classical-quantum channels \cite{AG26}, discussed in Sec.~\ref{sec: quantum}.

{\bf How the MRRW bound was improved:} 
 In a nutshell, the main idea can be expressed as the slogan ``move from packing vectors to packing subspaces, and move the subspace,'' in the sense discussed in detail in the main text.
 Interestingly, the same slogan applies to both proofs \cite{OAI26,AG26}---and, even more remarkably, the subspaces are closely related!
 
 More formally, the zonal vector at a point $x$ spans the unique line in each Fourier level fixed by the stabilizer of $x$. In the classical construction, one attaches to a codeword a weighted combination of these lines. The new construction replaces the trivial stabilizer representation by a nontrivial one, namely, the space of degree-$k$ Boolean harmonics, embedding an isometric copy of it into every retained level. This subspace moves equivariantly
 with the codeword, and it is of exponentially large dimension, accounting for the asymptotic improvement. At the functional level, the construction replaces the squared Christoffel–Darboux kernel formed of zonal vectors (equivalently, the trace kernel of the projections onto the attached lines) by the trace kernel of the projections onto the moving subspaces.

 \section{The result}\label{sec:result}

Given a code $C\subset \{0,1\}^n$, its {\em rate} and {\em distance} are defined as $R(C)=\frac{\log_2|C|}{n}$ and $d(C)=\min_{x,y\in C} d_H(x,y)$, where
the minimum is over all pairs of distinct vectors and $d_H$ denotes the Hamming distance. Let $A_2(n,d)=\max_{C:d(C)\ge d}|C|$ be the largest size of a code with distance $d=\delta n$. Let
   $$
 R_2(\delta)=\limsup_{n\to\infty}\tfrac1n\log_2A_2(n,\lceil\delta n\rceil)
   $$
be the largest attainable rate of a sequence of codes of growing length $n$ with {\em relative distance} $\delta$.
By the Gilbert--Varshamov bound,
   $$
   R_2(\delta)\ge 1-h_2(\delta), \quad 0<\delta<1/2.
   $$
The MRRW-1 bound further states that
   $$
   R_2(\delta)\le M_1(\delta):=h_2\Bigl(\tfrac12-\sqrt{\delta(1-\delta)}\Bigr).  
   $$

\vspace*{.2in}
The new bounds of \cite{OAI26} and \cite{AG26} are cited next.
\graybox
{{\em Theorem} (\cite[Eq.~(2)]{OAI26}) 
Let
   \begin{equation}\label{eq:OAI improved}
 \kappa_H(\delta)=\!\!\inf_{\substack{0\le b<a\le1/2\\ \Gamma_H(a,b)>1-2\delta}}\!\!
 \bigl(h_2(a)-h_2(b)\bigr),  \qquad
        \Gamma_H(a,b):=\frac{2(a-b)(1-a-b)}{\sqrt{a(1-a)}}.
   \end{equation}
Then $R_2(\delta)\le\kappa_H(\delta)<M_1(\delta)$ for every $\delta\in(0,1/2)$.}

\graybox
{{\em Theorem} (\cite[Theorem 2]{AG26}) Let
\begin{equation}\label{eq:AG-improved}
 R_{\mathrm{MQC}}(\delta)=\min_{0\le r\le\delta}
 \Bigl[\varsigma\bigl((1-2r)^2\bigr)-\varsigma\bigl(2(\delta-r)(1-2r)\bigr)\Bigr],
 \qquad \varsigma(D):=h_2\Bigl(\tfrac{1-\sqrt{1-D}}2\Bigr).
\end{equation}
Then $R_2(\delta)\le R_{\mathrm{MQC}}(\delta)<M_1(\delta)$ for every $\delta\in(0,1/2)$.}

First, we observe that these two results are identical:  
   \begin{equation}\label{eq: HS}
 R_{\mathrm{MQC}}(\delta)=\kappa_H(\delta),\qquad0<\delta<\tfrac12.
\end{equation}
Indeed, substitute $a=\frac12-\sqrt{r(1-r)}$ and $b=\bigl(1-\sqrt{1-2(\delta-r)(1-2r)}\bigr)/2$ in \eqref{eq:AG-improved}, so that the 
function under the minimum in \eqref{eq:AG-improved} becomes $h_2(a)-h_2(b)$. Writing $a=\frac12-\alpha$, $b=\frac12-\beta$, one has
$a-b=\beta-\alpha$, $1-a-b=\beta+\alpha$, $a(1-a)=(1-2r)^2/4$, and
\[
 4(\beta^2-\alpha^2)
 =1-2(\delta-r)(1-2r)-4r(1-r)
 =(1-2r)(1-2\delta),
\]
whence
\[
 \Gamma_H(a,b)=\frac{2(a-b)(1-a-b)}{\sqrt{a(1-a)}}
 =\frac{4(\beta^2-\alpha^2)}{1-2r}=1-2\delta.
\]
Thus, the MQC bound traverses the boundary curve $\Gamma_H(a,b)=1-2\delta$ of the OAI optimization problem \eqref{eq:OAI improved} with the same objective function. At the  endpoints we obtain MRRW-1 and the Bassalygo--Elias (BE) bound\footnote{Thus, the BE bound  and the MRRW-1 bound arise as two parameter settings of the same Delsarte certificate, which has not previously appeared within Delsarte's approach. Previous proofs of the BE bound using linear programming relied on a separate construction \cite[pp.~147-149]{Aaltonen1981} or on Rodemich's theorem \cite{delsarte1994application}.} \cite{Bassalygo1965}, \cite[p.93]{HuffmanPless2003}:
   \begin{equation}
\text{($r=\delta\leftrightarrow b=0$, MRRW--1);\qquad  ($r=0\leftrightarrow a=\frac12$,
BE)}.
  \end{equation}
Moreover, at $b=0$, we find that $\Gamma_H(a,0)=2\sqrt{a(1-a)}$, and the constraint becomes
the classical one since $\Gamma_H(a,0)> 1-2\delta$ is equivalent to $a>\frac12-\sqrt{\delta(1-\delta)}$. 
This implies that $\kappa_H(\delta)\le M_1(\delta)$ with equality only on the 
boundary\footnote{It is easily seen that the infimum in \eqref{eq:OAI improved} is generally attained for $b>0$.} $b=0$. 
It is also easily verified that for a fixed $\delta\in(0,1/2)$ the infimum in \eqref{eq:OAI improved} is attained 
on the boundary curve $\Gamma_H(a,b)=1-2\delta$, which implies the equality \eqref{eq: HS}.

The conclusion in \eqref{eq: HS} is not merely a coincidence; we expose the reasons in Sec.~\ref{sec: quantum} after introducing the harmonic-analytic framework where it fits properly.

{\footnotesize\begin{table}[h]
\centering
\begin{tabular}{c|cccccccc}
\toprule
 $\delta$ & 0.10 & 0.15 & 0.20 & 0.25 & 0.30 & 0.35 & 0.40 & 0.45 \\
\midrule
$\min(M_1(\delta),R_{\mathrm{BE}}(\delta))$
 & 0.701882 & 0.591857 & 0.468996 & 0.354579
& 0.250225 & 0.158133 & 0.0814689 & 0.0252661 \\
$R_{\mathrm{MQC}}(\delta)$
 & 0.699832 & 0.577921 & 0.460900 & 0.350379
& 0.248376 & 0.157506 & 0.0813366 & 0.0252574 \\
\bottomrule
\end{tabular}
\medskip
\caption{The MRRW-1 bound $M_1(\delta)=h_2\bigl(\frac12-\sqrt{\delta(1-\delta)}\bigr)$
and the improved bound \eqref{eq:OAI improved}-\eqref{eq:AG-improved}.
Note that for $\delta\le 0.0744$ the minimum in \eqref{eq:AG-improved} is attained 
for $r=0$, and $R_{\mathrm{MQC}}(\delta)=R_{\mathrm{BE}}(\delta)$.}
\end{table}
}

\section{Preliminaries}\label{sec:setup}
We begin with introducing the notation used in the constructions below. We will write the Hamming space in the multiplicative notation,  $X_n=\{\pm1\}^n$, and think of it as a group under coordinatewise multiplication. Its characters\footnote{\cite{OAI26} denotes characters both as $\chi_S$ and $e_S$, which creates some ambiguity.} have the form $\chi_S(z)=\prod_{r\in S}z_r, S\subset[n]$ (also called Walsh functions), and they form a basis of the space of functions $\R^{X_n}$. Under the translation action of $X_n$, this space decomposes into Fourier levels as 
   $$\R^{X_n}=\bigoplus_{i=0}^{n}V_i, \qquad  V_i=\operatorname{span}\{\chi_S:|S|=i\},
   $$
   where
   \[
D_i=\dim V_i=\binom ni.
\]
Below we use the normalized inner product on $\R^{X_n}$ given by
$\la f,g\ra={\mathbb E}_z\,f(z)g(z)=2^{-n}\sum_{z\in X_n}f(z)g(z)$, for which
the characters form an orthonormal basis.
We write $\chi_r$ for $\chi_{\{r\}}$, the $r$-th coordinate function, so $\chi_r(z)=z_r$ and
$\sum_r \chi_r= n t(\mathbf{1}, \cdot)$, where $t(x,y)$ is defined in \eqref{eq: distance}.

The isometry group of $X_n$ is $B_n=\{\pm1\}^n\rtimes S_n$; the $V_i$ are its isotypic components\footnote{they supply the global index, $i=an$, so also the parameter $a$.}, and the stabilizer of a point is $\Stab(x)\cong S_n$.

In the multiplicative notation, the distance turns into an inner product. Let $W=\R^n$ with standard basis 
$(e_1,\dots,e_n)$ and $\ell_x:=x/\sqrt n=n^{-1/2}\sum_r x_re_r$. Next, define
      \begin{equation}\label{eq: distance}   
 t(x,y):=\la\ell_x,\ell_y\ra=1-\frac{2\,d_H(x,y)}{n},\qquad s=1-2\delta ,
      \end{equation}
where $s$ is the corresponding relative parameter.

We will use a pair of adjoint operators, called \emph{raising and lowering}, defined on the basis as follows:
  $$
\partial\chi_S=\sum_{r\in S}\chi_{S\setminus\{r\}}, \qquad \partial^\ast\chi_S=\sum_{r\not\in S}
  \chi_{S\cup \{r\}}.
   $$
Note that
  \begin{equation}\label{eq: partial}
 \partial+\partial^{*}=\text{multiplication by }\textstyle\sum_r \chi_r; \quad 
 ((\partial+\partial^{*})f)(z)=(n-2\wt(z))f(z)
   \end{equation}
which is the operator whose compression is the tridiagonal matrix below.

\emph{Zonal vectors.} For $x\in X_n$ put
\begin{equation}\label{eq: zonal}
 v_{i,x}(z) = D_i^{-1/2}\sum_{|S|=i} \chi_S(x)\chi_S(z) \in V_i,
 \qquad \|v_{i,x}\|=1,\qquad v_{i,x}=R_xv_{i,\mathbf 1},
\end{equation}
where $R_x\chi_S=\chi_S(x)\chi_S$ is multiplication by the sign character.
Each $v_{i,x}$ spans a one-dimensional subspace, namely the line in $V_i$ fixed by $\Stab(x)$.

\vspace*{.1in}
\emph{Unwrapping the multiplier: Coordinate maps.} The proof in \cite{BN06} yields a general bound on the code size that involves the Perron eigenvalue of a tridiagonal matrix ($J_{i,i+1}^{(0)}$ below), which arises from a certain recurrence\footnote{We deliberately hide the details which do not matter here.} generated by the multiplication of the functions by $n-2\wt(z)$, where $\wt(\cdot)$ is the number of $-1$ entries in the argument. By \eqref{eq: partial}, this multiplication sums over the
coordinates; to record this action for each coordinate, define the mapping
  \begin{gather*}
 \nabla\colon\ \R^{X_n}\longrightarrow W\otimes\R^{X_n}, \qquad
 \nabla f=\Big(\frac1{\sqrt n}\sum_r e_r\otimes\chi_r\Big)f\\
(\nabla f)(z)
 =\Bigl[\frac1{\sqrt n}\sum_{r=1}^{n}e_r\otimes\chi_rf\Bigr](z)
=\frac1{\sqrt n}\sum_{r=1}^{n}(\chi_rf)(z)\,e_r
 =f(z)\,\ell_z 
 \end{gather*}
 
Since each coordinate function takes values $\pm1$, multiplication by $\chi_r$ is norm-preserving, and so $\nabla$ is an isometry. Let
\[
 (\ell_x\otimes\mathrm{id})^{*}\colon\ W\otimes\R^{X_n}\to\R^{X_n},
 \qquad
 (\ell_x\otimes\mathrm{id})^{*}(w\otimes f)=\la\ell_x,w\ra\,f
\]
be the contraction (partial inner product) in the first factor: it pairs the $W$-slot with
$\ell_x$ and leaves the function unchanged. Applying it to $\nabla$ gives the multiplier based at $x$:
\[
 (\ell_x\otimes\mathrm{id})^{*}\,\nabla f=t(x,\cdot)\,f .
\]
The reason for including the factor $W$ is that products $t(x,y)\la u,u'\ra$ are realized as inner products of vectors attached to $x$ and $y$:
\[
 \la\,\ell_x\otimes u,\ \ell_y\otimes u'\,\ra=t(x,y)\,\la u,u'\ra ,
\]
which is what will make the kernels below positive definite.

 Multiplication by $\sum_r\chi_r =n-2\wt(\cdot)$ results in a sum over the coordinates $1$ to $n$; we now keep the summands separate, recording the action of each
$z_r$ in its own `slot' of $W$. Throughout, we rely on the indexing
\[
 C_{i,j}\colon\ V_j\longrightarrow W\otimes V_i ,
\]
where the first index refers to the level of the target space. This mapping is
simpler than it looks because $\chi_r\chi_S=\chi_{S\triangle{r}}$ raises or lowers
the level by one, so only $j=i\pm1$ can ever occur.

\begin{lemma}[coordinate maps and the splitting of $\nabla$]\label{lem:coordmaps}
For $0\le i\le n$ define $C_{i,i+1}\colon V_{i+1}\to W\otimes V_i$ and
$C_{i,i-1}\colon V_{i-1}\to W\otimes V_i$ acting on basis vectors as
\begin{equation}\label{eq:Cdef}
\begin{aligned}
 C_{i,i+1}\chi_T&=\frac1{\sqrt{i+1}}\sum_{r\in T}e_r\otimes\chi_{T\setminus\{r\}}
 \quad(|T|=i+1),\\
  C_{i,i-1}\chi_{T}&=\frac1{\sqrt{n-i+1}}\sum_{r\notin T}e_r\otimes\chi_{T\cup\{r\}}
 \quad(|T|=i-1).
 \end{aligned}
\end{equation}
Then:
\begin{itemize}
\item[(i)] both maps are isometries, and their ranges in $W\otimes V_i$ are
orthogonal;
\item[(ii)]  for $f\in V_i$, the adjoints of these maps are given by\footnote{Formula \eqref{eq:Cadj} controls the way a function $f\in V_i$ is moved up or down one level. Summing it termwise over $\ell_x\otimes f=\sum_r\frac{x_r}{\sqrt n}\,e_r\otimes f$ gives rise to the transition coefficients below. It also ensures that multiplication by $t(x,y)$ appears separately as inner products for $x$ and $y$, giving rise to the Gram factorization and yielding positive definiteness in Cor.~\ref{cor:posdef}
.}
\begin{equation}\label{eq:Cadj}
 C^{*}_{i,i+1}(e_r\otimes f)=\frac{(\chi_rf)_{i+1}}{\sqrt{i+1}},
 \qquad
 C^{*}_{i,i-1}(e_r\otimes f)=\frac{(\chi_rf)_{i-1}}{\sqrt{n-i+1}},
\end{equation}
where $(\,\cdot\,)_j$ denotes the component in $V_j$;
\item[(iii)] $\nabla(V_i)\subset W\otimes(V_{i-1}\oplus V_{i+1})$ and
it splits into a lowering and a raising half\footnote{
The two square roots in \eqref{eq:nablasplit} are the \emph{masses} of the two terms in the recurrence. For a unit $f\in V_i$ the vector $\nabla f$ splits into orthogonal pieces of squared
norms
\[
 \Bigl\|\sqrt{\tfrac in}\,C_{i-1,i}f\Bigr\|^{2}=\frac in,
 \qquad
 \Bigl\|\sqrt{\tfrac{n-i}{n}}\,C_{i+1,i}f\Bigr\|^{2}=\frac{n-i}{n},
\]
because the $C$'s are isometries; the numbers $i$ and $n-i$ count the
coordinates that lower and raise the level, and the $\sqrt n$ is the
normalization of $\nabla$ itself. In particular
$\|\nabla f\|^2=\frac in+\frac{n-i}{n}=1$, confirming again that $\nabla$ is an
isometry, and the two masses are the transition probabilities $p_{i,i-1}$ and
$p_{i,i+1}$ that appear below.},
\begin{equation}\label{eq:nablasplit}
 \nabla|_{V_i}
 =\sqrt{\tfrac in}\;C_{i-1,i}\;+\;\sqrt{\tfrac{n-i}{n}}\;C_{i+1,i}.
\end{equation}
\end{itemize}
\end{lemma}

\begin{proof}
(i) The vectors $e_r\otimes\chi_S$, $r\in[n]$, $|S|=i$, form an orthonormal
basis of $W\otimes V_i$. In the first expression in \eqref{eq:Cdef} the pairs
$(r,T\setminus\{r\})$ with $r\in T$ are distinct, so the summands are
distinct basis vectors. There are $|T|=i+1$ of them, so
$\bigl\|\sum_{r\in T}e_r\otimes\chi_{T\setminus\{r\}}\bigr\|^2=i+1$,
and thus, $\|C_{i,i+1}\chi_T\|=1$. Moreover
$e_r\otimes\chi_{T\setminus\{r\}}=e_{r'}\otimes\chi_{T'\setminus\{r'\}}$
forces $r=r'$ and then $T=T'$, so images of distinct basis vectors are
orthogonal and $C_{i,i+1}$ is an isometry. The second sum has
$n-|T|=n-i+1$ terms and is proved in a similar way.

For the orthogonality of the ranges, note that $\operatorname{range}(C_{i,i+1})$
is spanned by tensors $e_r\otimes\chi_S$ with $r\notin S$, while
$\operatorname{range}(C_{i,i-1})$ is spanned by tensors with $r\in S$; the two
families are disjoint subsets of an orthonormal basis.

(ii) For $|T|=i+1$ and $|S|=i$,
\[
 \la C_{i,i+1}\chi_T,\;e_r\otimes\chi_S\ra
 =\frac1{\sqrt{i+1}}{\mathbbm 1}_{\{r\in T, S=T\setminus\{r\}\}}
 =\frac1{\sqrt{i+1}}{\mathbbm 1}_{\{r\notin S,T=S\cup\{r\}\}}
\]
Hence
\[
 C^{*}_{i,i+1}(e_r\otimes\chi_S)
 =\sum_{|T|=i+1}\la e_r\otimes\chi_S,\;C_{i,i+1}\chi_T\ra\,\chi_T=
 \dfrac{\chi_{S\cup\{r\}}}{\sqrt{i+1}}\mathbbm{1}_{\{r\notin S\}}.
 \]
As noted above, $\chi_r\chi_S=\chi_{S\triangle\{r\}}$, which lies in $V_{i+1}$ if $r\not\in S$ and
in $V_{i-1}$ if $r\in S$. Its component in $V_{i+1}$ is therefore $\chi_{S\cup\{r\}}\mathbbm{1}_{\{r\notin S\}}$, which yields the first identity in \eqref{eq:Cadj}. 
The second is obtained in the same way from the component in $V_{i-1}$ using
$\la C_{i,i-1}\chi_T,e_r\otimes\chi_S\ra=(n-i+1)^{-1/2}
{\mathbbm 1}_{\{r\in S,T=S\setminus\{r\}\}}$.

(iii) Let $|S|=i$. Write the sum defining $\nabla$ as two separate terms 
according to whether $r$ lies in $S$,
\[
 \nabla\chi_S
 =\frac1{\sqrt n}\sum_{r=1}^{n}e_r\otimes \chi_r\chi_S
 =\frac1{\sqrt n}\Bigl[\sum_{r\in S}e_r\otimes\chi_{S\setminus\{r\}}+
 \sum_{r\notin S}e_r\otimes\chi_{S\cup\{r\}}\Bigr]
,
\]
where the first sum collects the terms in $W\otimes V_{i-1}$ and the second
in $W\otimes V_{i+1}$.
Applying \eqref{eq:Cdef} with $i$ replaced by $i-1$ in the first formula
(so that the target is $W\otimes V_{i-1}$ and the normalizing factor is
$|S|^{-1/2}=i^{-1/2}$), and with $i$ replaced by $i+1$ in the second
(target $W\otimes V_{i+1}$, factor $(n-|S|)^{-1/2}=(n-i)^{-1/2}$), gives
\[
 \sum_{r\in S}e_r\otimes\chi_{S\setminus\{r\}}=\sqrt i\;C_{i-1,i}\chi_S,
 \qquad
 \sum_{r\notin S}e_r\otimes\chi_{S\cup\{r\}}=\sqrt{n-i}\;C_{i+1,i}\chi_S,
\]
which is \eqref{eq:nablasplit}.
\end{proof}

\section{The classical bound: one line per codeword}\label{sec:bn06}

In this section, we recall the spectral proof of MRRW-1 in \cite{BN06}, with a minor adjustment. 
The original proof was written in the ``radial form'', wherein the functions on the space $X^n$ are collapsed into functions
of the Hamming weight\footnote{Most classical works relied on radial functions \cite{MRRW77,KabatyanskiiLevenshtein1978} while more recent proofs of MRRW-1 \cite{FT05,NS09} also used functions on the whole space.}. 
Here we rewrite the proof, unwrapping the functions into functions on the entire space\footnote{We will present the proof in an ``LP-free'' form, without mentioning explicitly linear programming or orthogonal polynomials, which will remain in the background. While the argument in \cite{BN06} is more transparent, our goal is the proof in \cite{OAI26} which apparently cannot be written using radial functions.}. While here the unwrapping trick is not necessary, it will become essential when we move to the new result, the bound \eqref{eq:OAI improved}.
The bulk of the technical work is moved into this rewriting; once this is done, the new bound \eqref{eq:OAI improved} follows readily.

\subsection{Transition coefficients}
Fix $L<n/2$ and define
  $$
  V=\bigoplus_{i=0}^{L}V_i, \qquad D=\dim V=\sum_{i\le L}D_i.
  $$
The next lemma is the centerpiece of the derivation.

\begin{lemma}[transition coefficients]\label{lem:contract0}
For every $x$ and every $i$,
\[
 C^{*}_{i,i+1}\bigl(\ell_x\otimes v_{i,x}\bigr)=\sqrt{p_{i,i+1}}\;v_{i+1,x},
 \qquad
 C^{*}_{i,i-1}\bigl(\ell_x\otimes v_{i,x}\bigr)=\sqrt{p_{i,i-1}}\;v_{i-1,x},
\]
with
\[
 p_{i,i+1}=\frac{n-i}{n},\qquad p_{i,i-1}=\frac{i}{n},
 \qquad p_{i,i+1}+p_{i,i-1}=1 .
\]
\end{lemma}

\begin{proof}
By equivariance take $x=\mathbf 1$. From the definitions,
$C^{*}_{i,i+1}(e_r\otimes\chi_S)=(i+1)^{-1/2}\chi_{S\cup\{r\}}$ for
$r\notin S$ and $0$ otherwise, so
$C^{*}_{i,i+1}(\ell\otimes v_{i,\mathbf 1})
=\bigl(n(i+1)\bigr)^{-1/2}\,\partial^{*}v_{i,\mathbf 1}$; and
$\partial^{*}v_{i,\mathbf 1}=(i+1)\sqrt{D_{i+1}/D_i}\;v_{i+1,\mathbf 1}$.
Now use $D_{i+1}/D_i=(n-i)/(i+1)$. The proof of the second identity is fully analogous.
\end{proof}

\begin{remark}\label{rem:ehrenfest}
The numbers $p_{i,i\pm1}$ are the transition probabilities of the Ehrenfest
urn chain on the levels. The dimension balance
$D_ip_{i,i+1}=D_{i+1}p_{i+1,i}$ is its detailed balance with respect to the
binomial stationary measure. 
\end{remark}

\subsection{The field of lines}

Let $J^{(0)}=J^{(0)}(n,L)$ be the 
$(L+1)\times(L+1)$ symmetric tridiagonal matrix with zero diagonal and
\[
 J^{(0)}_{i,i+1}=\sqrt{p_{i,i+1}\,p_{i+1,i}}=\frac{\sqrt{(i+1)(n-i)}}{n}.
\]

\begin{remark}[the two halves reassemble into a recurrence]\label{rem:recur}
Applying $(\ell_x\otimes\mathrm{id})^{*}$ to $\nabla$ and using
Lemma~\ref{lem:contract0} recovers multiplication by the distance
coordinate:
\[
 t(x,\cdot)\,v_{i,x}
 =\alpha_i\,v_{i+1,x}+\alpha_{i-1}\,v_{i-1,x},
 \qquad
 \alpha_i=\frac{\sqrt{(i+1)(n-i)}}{n}=J^{(0)}_{i,i+1},
\]
so $J^{(0)}$ is the matrix of this three-term recurrence, and the coordinate
maps are its raising and lowering halves. Written in the radial form, this is the classical recurrence for
Krawtchouk polynomials. For $k>0$ the two halves
no longer reassemble in this way, and it is there that the
Christoffel--Darboux structure is lost.
\end{remark}

Let $\Lambda=\lmax(J^{(0)})$ and let $\theta=(\theta_i)$ be its unit Perron eigenvector,
$\theta_i>0$. Put
\begin{equation}\label{eq: Perron}
 w_i=\sqrt{D_i}\,\theta_i,\qquad Z=\sum_{i\le L}w_i,
 \qquad\text{so that}\qquad
 \sum_{i}p_{i,j}w_i=\Lambda w_j \quad (0\le j\le L).
\end{equation}
Finally attach to each codeword the unit vector
\[
 u_x=\sum_{i=0}^{L}\sqrt{\tfrac{w_i}{Z}}\;v_{i,x}\in V,
 \]
let $P_x=u_xu_x^{T}$ be the orthogonal projector on $u_x$ and define
  \[
 K(x,y)=\tr(P_xP_y)=\la u_x,u_y\ra^{2}.
\]
A trivial but useful observation is that $K\ge 0$. 

{\em The field of lines}. The field $x\mapsto\R u_x$ is equivariant for the group $B_n$, in the sense that for the orthogonal representation $\rho$ of $B_n$ on $\R^{X_n}$ we have $\rho(g)u_x=u_{gx}$ for all $g\in B_n$. The same holds verbatim in Section~\ref{sec: quantum}, with 
$\rho$ the unitary representation of $B_n$ on $(\C^2)^{\otimes n}$ and $u_x$ replaced by the output of the PSC. Since
$B_n$ is transitive on ordered pairs of points at a given distance, it follows
that $K(x,y)$ depends only on $d_H(x,y)$ but not on the points.

\subsection{The certificate}

Positivity of the certificate relies only on the transition coefficients
and the Perron relation, and every step is agnostic of the harmonic degree, which will be reused
in Section~\ref{sec:enlarge} when we turn to the generalization. The
original certificate of \cite{BN06} is slightly different and yields a
(nonexponentially) sharper bound; it is briefly discussed in Remark~\ref{rem: LP} and in Section~\ref{sec: complements}.

\begin{lemma}[the isometry]\label{lem:B}
Define $B\colon V\to\R^n\otimes V$ by
\[
 Bh=\bigoplus_{i}\ \sum_{j:|i-j|=1}
 \sqrt{\frac{p_{i,j}\,w_i}{\Lambda\,w_j}}\;C_{i,j}h_j ,
 \qquad h=\bigoplus_jh_j .
\]
Then $B^{*}B=\mathrm{id}_V$ and $B^{*}(\ell_x\otimes u_x)=\sqrt\Lambda\,u_x$
for every $x$.
\end{lemma}

\begin{proof}
For fixed $i$ the two summands have orthogonal ranges, so
$\|Bh\|^2=\sum_j\|h_j\|^2\,(\Lambda w_j)^{-1}\sum_ip_{i,j}w_i=\|h\|^2$ by the
Perron relation \eqref{eq: Perron}. For the second claim, the $j$-th component of
$B^{*}(\ell_x\otimes u_x)$ equals, by Lemma~\ref{lem:contract0},
\[
 \sum_i\sqrt{\frac{p_{i,j}w_i}{\Lambda w_j}}\sqrt{\frac{w_i}{Z}}\sqrt{p_{i,j}}
 \;v_{j,x}
 =\frac{1}{\sqrt{\Lambda w_jZ}}\Bigl(\sum_ip_{i,j}w_i\Bigr)v_{j,x}
 =\sqrt{\Lambda}\,\sqrt{\tfrac{w_j}{Z}}\;v_{j,x}. \qedhere
\]
\end{proof}

\begin{corollary}[positivity]\label{cor:posdef} The kernel $(t-\Lambda)K$ is positive-definite.
\end{corollary}
\begin{proof} Denote $\Theta_x:=(\ell_x\otimes\mathrm{id})P_x-\sqrt\Lambda\,BP_x$, then
\[
 \bigl(t(x,y)-\Lambda\bigr)K(x,y)=\la\Theta_x,\Theta_y\ra_{\mathrm {HS}}. 
\]
Indeed, 
\begin{gather*}
\la (\ell_x\otimes\mathrm{id})P_x,(\ell_y\otimes\mathrm{id})P_y\ra_{\mathrm {HS}}=t(x,y) K(x,y),\\
\la BP_x, BP_y\ra_{\mathrm {HS}}=K(x,y),\qquad\la (\ell_x\otimes\mathrm{id})P_x,BP_y\ra_{\mathrm{HS}}=
\sqrt\Lambda K(x,y)
\end{gather*}
where the last two equalities follow by Lemma~\ref{lem:B}.
\end{proof}

\subsection{The bound}

\begin{theorem}[\cite{BN06} in full-space form]\label{thm:bn06}
Let $C\subseteq X_n$ have minimum distance at least $\delta n$, and suppose
$\Lambda>s$, where $s=1-2\delta$; Eq.~\eqref{eq: distance}. Then
\[
 |C|\;\le\;\frac{1-s}{\Lambda-s}\;D .
\]
\end{theorem}

\begin{proof}
We have
  \begin{equation}\label{eq:diag}
  \sum_{x,y\in C}(t(x,y)-s)K(x,y)\le \sum_x(t(x,x)-s)K(x,x)=|C|(1-s)
  \end{equation}
Further, let $\widehat K_0= \E_{x,y}K(x,y)$ be the average of $K$ on $X_n\times X_n$, then\footnote{The computation in \eqref{eq:mean} is the Delsarte inequality about the mean in the form
used by Kabatiansky and Levenshtein \cite{KabatyanskiiLevenshtein1978}.}
  \begin{equation}\label{eq:mean}
 \sum_{x,y\in C}K(x,y)
 =\sum_{S}\widehat K(S)\Bigl|\sum_{x\in C}\chi_S(x)\Bigr|^{2}
 \;\ge\;\widehat K_0\,|C|^{2}.
\end{equation}
The last inequality uses $\hat K(S)\ge 0$ for all $S$.
Let $\Pi=\E_x[P_x]$, then 
$\widehat K_0=\tr(\Pi^2)$, and since $\tr\Pi=1$ and $\Pi$ has rank at most
$D$,
\begin{equation}\label{eq:CS}
 \widehat K_0=\tr(\Pi^{2})\stackrel{\text{C.-S.}}{\ge}\frac{(\tr\Pi)^{2}}{D}=\frac1D .
\end{equation}
Finally, using Cor.\ref{cor:posdef} and \eqref{eq:CS}
    \begin{align}\label{eq: proj}
    \sum_{x,y\in C}(t-s)K&=\sum_{x,y\in C}(t-\Lambda)K+\sum_{x,y\in C}(\Lambda-s)K 
    \ge (\Lambda-s)\frac {|C|^2}D.
    \end{align}
Combining \eqref{eq:diag}-\eqref{eq: proj} yields the result.
\end{proof}
\begin{remark}
Steps \eqref{eq:mean}-\eqref{eq:CS} can be jointly written as
  \begin{equation}\label{eq: two steps}
   \sum_{x,y\in C}K=\|\sum_xP_x\|_{HS}^2\ge\Big(\tr\sum_xP_x\Big)^2/D.
   \end{equation}
The proof of this theorem can be simplified for the case of lines, but in this form it generalizes directly to the case of subspaces.    
\end{remark}

\begin{remark}\label{rem: LP}
Using Cauchy-Schwarz in \eqref{eq:CS} entails a loss, which can be avoided by evaluating $\tr(\Pi^{2})$ rather than bounding it. We have ${\mathbb E}_x[v_{i,x}v_{j,x}^T]=D_i^{-1}{\mathrm id}_{V_i}{\mathbbm 1}_{\{i= j\}}$ and thus $\widehat K_0=\tr{\Pi^2}=Z^{-2}$, where $Z$ is defined in \eqref{eq: Perron}. Since $Z\le \sqrt{D}$, this gives a slightly (non-exponentially) sharper inequality $|C|\le\frac{1-s}{\Lambda -s}Z^2$, which is the bound in \cite{BN06}.
The same computation in Section~\ref{sec:enlarge} gives
$\operatorname{tr}\Pi=d_E$, $\widehat K_0=d_E^{2}/Z^{2}$ and
$|C|\le\frac{1-s}{\Lambda-s}\cdot\frac{Z^{2}}{d_E}$.
\end{remark}

\begin{remark}[$J^{(0)}$ is the adjacency matrix of a ball]\label{rem:ball}
From a vertex of weight $j$ the cube has $n-j$ neighbours of weight $j+1$
and $j$ of weight $j-1$. These up- and down-degrees satisfy the detailed
balance $\binom nj(n-j)=\binom n{j+1}(j+1)$ with respect to the binomial
measure, so the associated symmetric matrix has entries
$\sqrt{(j+1)(n-j)}$, and
\[
 n\,J^{(0)}_{j,j+1}=\sqrt{(j+1)(n-j)} .
\]
Thus $nJ^{(0)}$ is exactly the symmetrized quotient --- the radialization ---
of the adjacency matrix of the induced subgraph on the ball
$B(L)=\{z\colon \wt(z)\le L\}$; equivalently, of the radial adjacency of the
cube truncated to the first $L+1$ weights.
\end{remark}

\section{The improved bound : one subspace per codeword}\label{sec:enlarge}
For $0\le k\le n/2$ define the subspace\footnote{
$\cH_k=\{\sum_{|S|=k} a_S\chi_S:\sum_{S\supseteq T}a_S=0\}$ for every $T$ with $|T|=k-1$ \cite[Sec.II]{Schrijver05}. This is the space of higher-degree Boolean harmonics, which can be described in representation-theoretic terms, but we will not need this here.} 
\[
 \cH_k=\ker\bigl(\partial\colon V_k\to V_{k-1}\bigr), 
 \qquad
 d_E:=\dim\cH_k=\binom nk-\binom n{k-1},
\]
embedded isometrically in each Fourier level $V_i, k\le i\le n-k$ by
  $$
     \varphi_i h=\frac{(\partial^{*})^{\,i-k}h}
 {(i-k)! \binom{n-2k}{i-k}^{1/2}}=\binom{n-2k}{i-k}^{-1/2}\sum_{|T|=i}\Big(\sum_{S\subseteq T, |S|=k}a_S\Big)\chi_T
  $$
where $h=\sum_{|S|=k} a_S\chi_S$, and
  $$
  \varphi_{i,x}=R_x\varphi_i,\quad R_x\chi_S:=\chi_S(x)\chi_S .
  $$
The classical case of Sec.~\ref{sec:bn06} corresponds to $k=0$, then the space $\cH_0$ is the constants and 
$\varphi_{i,x}$ is
the zonal vector $v_{i,x}$.

\subsection{The tridiagonal matrix, again}
Let $Q_i=(i-k+1)(n-i-k), k\le i<n/2$. We have
  $$
  \partial^{*}\varphi_ih=\sqrt{Q_i}\,\varphi_{i+1}h, \qquad
\partial\varphi_{i+1}h=\sqrt{Q_i}\,\varphi_ih,
$$
and Lemma~\ref{lem:contract0}
becomes
\[
 C^{*}_{i,i+1}\bigl(\ell_x\otimes\varphi_{i,x}h\bigr)=\sqrt{p_{i,i+1}}\,\varphi_{i+1,x}h,
 \qquad
 C^{*}_{i,i-1}\bigl(\ell_x\otimes\varphi_{i,x}h\bigr)=\sqrt{p_{i,i-1}}\,\varphi_{i-1,x}h,
\]
now\footnote{The new feature that arises is that the transition coefficients no longer sum to one (are ``sub-stochastic''), so the associated Markov chain is no longer conservative, as opposed to the case $k=0$; cf. Lemma~\ref{lem:contract0}. This is the meaning of the remark about coordinate multiplication reaching ``representations outside the selected path'' in \cite[Sec.2.1]{OAI26}.}
   \begin{gather}
 p_{i,i+1}=\frac{Q_i}{n(i+1)},\qquad
 p_{i,i-1}=\frac{Q_{i-1}}{n(n-i+1)},
 \\
 D_i\,p_{i,i+1}=D_{i+1}\,p_{i+1,i},
 \\
 p_{i,i+1}+p_{i,i-1}<1\ \ (k>0). \label{eq: nonstochastic}
 \end{gather}
The matrix $J^{(k)}=J^{(k)}(n,k,L)$ has zero diagonal and
\[
 J^{(k)}_{i,i+1}=\sqrt{p_{i,i+1}p_{i+1,i}}
 =\frac{(i-k+1)(n-i-k)}{n\sqrt{(i+1)(n-i)}},
 \qquad i=k,\dots,L-1 .
\]

\noindent Observe that the case $k=0$ recovers the matrix $J_{i,i+1}^{(0)}$. 

\subsection{The moving subspace and the bound}

Let $\theta$ be the unit Perron eigenvector of $J^{(k)}$, $\Lambda=\lmax(J^{(k)})$,
and, exactly as before in \eqref{eq: Perron} and Lemma~\ref{lem:B},
   \begin{gather}\nonumber
        w_i=\sqrt{D_i}\,\theta_i,\quad Z=\sum_{i=k}^{L}w_i,\\[-.1in]
 \Psi_xh=\bigoplus_{i=k}^{L}\sqrt{\tfrac{w_i}{Z}}\,\varphi_{i,x}h,
 \qquad
 E_x=\operatorname{im}\Psi_x,\quad P_x=\Psi_x\Psi_x^{*}. \label{eq: tower}
   \end{gather}
Thus $\Psi_x\colon\cH_k\to V$ and $\dim E_x=d_E$, where now
\begin{equation}\label{eq: V}
 V=\bigoplus_{i=k}^{L}V_i,\qquad D=\dim V=\sum_{i=k}^{L}D_i .
\end{equation}
The ambient space depends on $k$ only through the lower endpoint of the range,
and since $L\le n/2$ the sum is dominated by its top term, so discarding the
levels below $k$ changes $D$ by a subexponential factor only.

{\em The field of subspaces.} What is attached to $x$ is now the
$d_E$-dimensional subspace $E_x=\operatorname{im}\Psi_x\subset V$ rather than a
line, and it moves with the codeword: with $\rho$ as in
Section~\ref{sec:bn06},
\[
 \rho(g)E_x=E_{gx},\qquad\text{equivalently}\qquad
 P_{gx}=\rho(g)P_x\rho(g)^{*}\qquad(g\in B_n).
\]

\graybox{Equivariance is checked as before except that it now forms a statement about the subspace
rather than an individual vector, the
permutation part acting on $E_x$ through the nontrivial $\Stab(x)$-representation 
$\cH_k$ (the source of the local index $b$); this is the precise content of the slogan
``replace lines with moving subspaces''.  Consequently, $K(x,y)=\tr(P_xP_y)$ is $B_n$-invariant and thus a function of
$d_H(x,y)$ alone, and the certificate stays {\em two-point} and {\em scalar}. This is the nexus of the new contribution in \cite{OAI26}, and it deserves being highlighted. }

\begin{theorem}[the projection bound]\label{thm:proj}
If $\Lambda>s$ then every code with minimum distance at least $\delta n$
satisfies
  \begin{equation}\label{eq: finite bound}
 |C|\;\le\;\frac{1-s}{\Lambda-s}\cdot\frac{D}{d_E},
 \qquad D:=\sum_{i=k}^{L}D_i .
 \end{equation}
\end{theorem}

\begin{proof}
We can reuse the proof of Theorem~\ref{thm:bn06} since Lemma~\ref{lem:B} and
Corollary~\ref{cor:posdef} do not depend on $k$ and therefore still hold 
(they used only the identities of Lemma~\ref{lem:contract0} and the Perron relation). The only difference is that we now use
$\tr\sum_{x\in C}P_x=|C|d_E$ and the fact that the diagonal elements of $K$ equal $d_E$ rather than 1, implying that
$\sum_{x,y\in C}K\ge|C|^2d_E^2/D$. Collecting the calculations, we have
  \begin{equation}\label{eq: mean}
  \frac{|C|^2d_E^2}D\le \sum_{x,y\in C} K(x,y)=\sum_{x,y\in C}\tr(P_x P_y)\le \frac{1-s}{\Lambda -s}|C|d_E. \qedhere
  \end{equation}
\end{proof}

\begin{remark}[packing of subspaces]\label{rem: packing}
    This theorem and in particular, Eq.~\eqref{eq: mean}, can be interpreted as a packing-of-subspaces argument. We attach a $d_E$-dimensional subspace to each code point and the ambient space has dimension $D$. Had these subspaces been pairwise orthogonal,
    their number would be bounded by $D/d_E$ (this is the argument behind the Welch bound for line packings
    and its extension to fusion frames in higher dimensions); inequalities in \eqref{eq: mean} bound
    the number of subspaces associated with points at distance $\ge d$ which are nearly orthogonal, and the
    factor $\frac{1-s}{\Lambda-s}$ accounts for the overlap.

    This remark extends to the one-dimensional case, covered in Theorem~\ref{thm:bn06}: for $k=0$ we obtain $\sum_{x,y}\la u_x,u_y\ra^2\ge |C|^2/D$, which yields the Welch bound.        
\end{remark}

\section{Asymptotics}\label{sec:asympt}

\begin{lemma}\label{lem:eig} Let $n\to\infty$ and let
$k/n\to b$ and $L/n\to a$ with $0\le b<a\le1/2$, then
\[
 \lmax\bigl(J^{(k)}(n,k,L)\bigr)\longrightarrow\Gamma_H(a,b)
 =\frac{2(A-B)}{\sqrt A},
\]
where
  \begin{equation}\label{eq: AB}
  A=a(1-a),\qquad B=b(1-b),\qquad
A-B=(a-b)(1-a-b).
  \end{equation}
\end{lemma}

\begin{proof} The argument is standard (e.g., \cite[Lemma 2.3]{BN06}).
    Fix $0<k<L<n/2$. The matrix elements of $J_{i,i+1}^{(k)}$ increase on $i=k,\dots,L-1$, so the maximum is attained for $i=L-1$. By Perron-Frobenius, $\lmax$ is bounded above
    by the maximum row sum. For the lower bound, Rayleigh-Ritz gives 
    $\lmax \ge {u^T J^{(k)}u}$ for an arbitrary unit vector $u$. Taking $u$ with 
 the last $m=o(n)$ coordinates $\frac1{\sqrt m}$ and $0$ otherwise yields the matching asymptotics.
\end{proof}

\begin{proof}[Proof of the bound \eqref{eq:OAI improved}]
Let $k/n\to b, L/n\to a$ and suppose that $\Gamma_H(a,b)>1-2\delta$. By Lemma~\ref{lem:eig} the
hypothesis $\Lambda>s$ of Theorem~\ref{thm:proj} holds for all large $n$. Using standard asymptotics
of binomial coefficients, $\frac 1n\log_2\frac D{d_E}\to h_2(a)-h_2(b)$, which is exactly Eq.~\eqref{eq:OAI improved}.\qedhere

\end{proof}

\section{The probabilistic picture: Classical-quantum channels}\label{sec: quantum}
A concurrent proof of the improved bounds on binary codes  was presented in the preprint \cite{AG26}, appearing simultaneously with \cite{OAI26}. The two proofs, different as they look, are both dimension counts for packings of subspaces that move with the codeword; what differs is how the subspaces are produced and in what sense they are packed. In this section we give a brief explanation, highlighting the differences and similarities. 
All the results below are due to \cite{AG26}, although we rewrite them from a different perspective, highlighting their geometric contents.

We will construct analogs of the field of lines, the space $V$, and the field of subspaces
of Section~\ref{sec:enlarge} and point out that the ``hard objects'' of the earlier part correspond to {\em typical objects} in the probabilistic picture. 

Define the state $|\mathrm{Ber}(p)\ra=\sqrt{1-p}\,|0\ra+\sqrt p\,|1\ra$. In \cite{AG26} the pure-state channel (PSC) sends a codeword $c$ to the unit vector
    \[
 |\sigma_c\ra=X^{c}\,|\mathrm{Ber}(p)^{n}\ra ,
 \qquad
  |\mathrm{Ber}(p)^n\ra=\sum_{e\in\{0,1\}^n}\sqrt{p^{\wt(e)}(1-p)^{n-\wt(e)}} |e\ra,
   \]
i.e., it attaches to each codeword a \emph{line} 
$\C|\sigma_c\ra\subset (\C^{2})^{\otimes n}$, with rank-one projector
$P_c=|\sigma_c\ra\la\sigma_c|$, which is the counterpart of the field of lines
$\R u_x$ of Section~\ref{sec:bn06}. Since the state is a product and
   $$
   \la\mathrm{Ber}(p)|X|\mathrm{Ber}(p)\ra=2\sqrt{p(1-p)},
   $$
   coordinates where
$x$ and $y$ agree contribute $1$ and coordinates where they differ
contribute $\gamma:=2\sqrt{p(1-p)}$, so
\[
 \la\sigma_x|\sigma_y\ra=\gamma^{\,d(x,y)} ,
\]
which is a function of the Hamming distance between the two points. Moreover, the family $\{\C|\sigma_x\ra\}$ is an equivariant field of lines under the full isometry group $B_n$ in the sense of Section~\ref{sec:bn06}. Namely, for the representation $(\C^2)^{\otimes n}$
\begin{itemize}
    \item a translation $\epsilon_y$ acts on $X^y:=X^{y_1}\otimes\dots\otimes X^{y_n}$, and
    $X^y|\sigma_x\ra=X^y X^x|\mathrm{Ber}(p)^{n}\ra=|\sigma_{x+y}\ra$;
    \item A permutation $\tau\in S_n$ acts by the unitary $U_\tau$ that permutes the
tensor factors. It satisfies $U_\tau X^{x}U_\tau^{*}=X^{\tau x}$, and it fixes
the base state, $U_\tau|\mathrm{Ber}(p)^n\ra=|\mathrm{Ber}(p)^n\ra$, which is 
a product of $n$ identical factors. Hence
\[
 U_\tau|\sigma_x\ra=U_\tau X^{x}|\mathrm{Ber}(p)^n\ra
 =X^{\tau x}U_\tau|\mathrm{Ber}(p)^n\ra=|\sigma_{\tau x}\ra,
\]
\end{itemize}
so again $\rho(g)|\sigma_x\ra=|\sigma_{gx}\ra$.

The ambient space for this line family is identified by the spectrum of
the average output state. For a single bit $\beta\in\{0,1\}$ write
\[
 |\sigma_0\ra=|\mathrm{Ber}(p)\ra,\qquad
 |\sigma_1\ra=X|\mathrm{Ber}(p)\ra,\qquad
 \sigma_\beta=|\sigma_\beta\ra\la\sigma_\beta| ,
\]
the density matrices $\sigma_\beta$ being the case $n=1$ of the projectors $P_c$
above. 
For a uniform input the average single-letter output is
\begin{equation}\label{eq: sigmabar}
 \bar\sigma=\tfrac12\bigl(\sigma_0+\sigma_1\bigr)
 =\begin{pmatrix}\frac12&\sqrt{p(1-p)}\\[3pt]
 \sqrt{p(1-p)}&\frac12\end{pmatrix}
 =\tfrac12 I+\sqrt{p(1-p)}\,X .
\end{equation}
The state $\bar\sigma$ is a real combination of $I$ and $X$, so it is diagonal in
the eigenbasis of $X$, that is, in the Hadamard basis
$|\pm\ra=(|0\ra\pm|1\ra)/\sqrt2$:
\begin{equation}\label{eq: q}
 \bar\sigma=(1-q)\,|+\ra\la+|\;+\;q\,|-\ra\la-| ,
 \qquad \text{where }q:=\tfrac12-\sqrt{p(1-p)}\in\bigl[0,\tfrac12\bigr].
\end{equation}
Thus the output ensemble, which is not diagonal in the computational basis,
becomes diagonal after a Hadamard rotation, with a Bernoulli$(q)$ spectrum.

Passing to $n$ letters, \eqref{eq: q} gives
   \begin{gather*}
 \bar\sigma^{\otimes n}
 =\sum_{S\subseteq[n]}q^{|S|}(1-q)^{\,n-|S|}\;|h_S\ra\la h_S| ,\\ 
 |h_S\ra:=H^{\otimes n}|S\ra
 =\bigotimes_{r\in S}|-\ra\ \bigotimes_{r\notin S}|+\ra ,  \quad S\subset[n]     
   \end{gather*}
where $|S\ra$ is the computational basis vector indexed by the indicator of
$S$. Under the identification $(\C^2)^{\otimes n}\cong\C^{\{0,1\}^n}$, 
we obtain
$|h_S\ra=2^{-n/2}\sum_{c}(-1)^{\la S,c\ra}|c\ra$, so $|h_S\ra$ \emph{is} the
character $\chi_S$ of Section~\ref{sec:setup} (the factor $2^{-n/2}$ 
arises because we switch from the counting measure to the uniform one for which the
characters are unit vectors). Consequently the eigenvalue of
$\bar\sigma^{\otimes n}$ at $\chi_S$ depends on $S$ only through $|S|$, and the
eigenspaces of $\bar\sigma^{\otimes n}$ are the (complexified) Fourier levels,
\[
 \bar\sigma^{\otimes n}\big|_{V_i}=q^{i}(1-q)^{\,n-i}\cdot\mathrm{id},
 \qquad \dim V_i=D_i=\binom ni .
\]

The spectrum of $\bar\sigma^{\otimes n}$ is therefore the binomial distribution
with parameter $q$, spread over the Fourier levels, and its von Neumann entropy equals ${\sf S}(\bar\sigma)=h_2(q)$ (e.g., \cite[p.253]{Wilde17}). Its typical (high-probability) subspace is spanned by the levels of degree close to $qn$ and has dimension $2^{(h_2(q)+o(1))n}$. Setting $a=q$, this is the ambient space $V$ of Section~\ref{sec:bn06}, with a minor difference between the ball and a spherical shell close to its surface, which does not matter for the exponential asymptotics, 
so the dimension
$\dim(V)=2^{(h_2(a)+o(1))n}$. 

The parameter $q$ and the syndrome duality statement $H^{\otimes n}|\mathrm{Ber}(p)^n\ra=|\mathrm{Ber}(q)^n\ra$  form the contents of \cite[\S1.5, Prop.~5]{AG26}, so we are not adding anything to their conclusions bringing the ambient space of Section~\ref{sec:bn06} into view on the quantum side.

Moving from lines to subspaces is accomplished by moving from pure states to the mixed ones, and 
this is the step taken in \cite[\S5.1]{AG26}: the {\em mixed-qubit
channel} (MQC) is the PSC followed by the $X$-Pauli channel
$\mathcal B_\eta(\varrho)=(1-\eta)\varrho+\eta X\varrho X$, $\eta\in[0,\frac12]$.
Its single-qubit outputs are
\begin{equation}\label{eq: mqc}
 \varrho_\beta=(1-\eta)\sigma_\beta+\eta\,\sigma_{\beta\oplus1}
 =\tfrac12\Bigl[I+\gamma\,X+(-1)^{\beta}(1-2p)(1-2\eta)\,Z\Bigr],
 \qquad \beta\in\{0,1\},
\end{equation}
and at blocklength $n$ the output attached to a codeword is the mixture
\[
 \varrho_c=\E_{e\sim\mathrm{Ber}(\eta)^n}\bigl[\,P_{c\oplus e}\,\bigr]
 =\bigotimes_{r=1}^{n}\varrho_{c_r} ,
\]
the error pattern $e$ being generated and then discarded; it is this discarding
that produces the mixedness \cite[Footnote 9]{AG26}. Here $p$ is the parameter
denoted by $r$ in \eqref{eq:AG-improved}.

Two features of \eqref{eq: mqc} are worth isolating, because they reproduce the
two-tier structure of Section~\ref{sec:enlarge}. First, the average output is
unchanged,
\[
 \bar\varrho=\tfrac12(\varrho_0+\varrho_1)
 =\tfrac12 I+\sqrt{p(1-p)}\,X=\bar\sigma ,
\]
since the $Z$-terms cancel. The ambient space is therefore the same as for the
PSC, namely, the typical subspace of $\bar\sigma^{\otimes n}$, of dimension
$2^{(h_2(a)+o(1))n}=D$ with $a=q=\frac12-\sqrt{p(1-p)}$. Observe that mixing does not
move it, exactly as the ambient dimension $D$ of Section~\ref{sec:enlarge} is fixed once we choose $a$,
and is the same for every harmonic degree $k$ up to a subexponential factor.
Second, the two outputs are distinguished precisely by the $Z$-term,
whose sign flips with $\beta$. In Bloch coordinates the states \eqref{eq: mqc} have
vectors \cite[p.31]{AG26}
\[
 v_\beta=\bigl(\gamma,\,0,\,(-1)^{\beta}(1-2p)(1-2\eta)\bigr),
\]
with the common $X$-component (which is the average, hence the ambient space) and
$Z$-components of opposite signs (which is what moves with the codeword). Write $\vec\sigma:=(X,Y,Z)$. Since
$\varrho_\beta=\frac12(I+v_\beta\cdot\vec\sigma)$ has eigenvalues
$\frac12(1\pm\|v_\beta\|)$ and
\[
 1-\|v_\beta\|^{2}=1-4p(1-p)-(1-2p)^{2}(1-2\eta)^{2}=4\eta(1-\eta)(1-2p)^{2},
\]
the entropy of a bit-dependent output is ${\sf S}(\varrho_\beta)=h_2(b)$ with
\begin{equation}\label{eq: b-mqc}
 b=\tfrac12\Bigl(1-\sqrt{1-4\eta(1-\eta)(1-2p)^{2}}\Bigr) .
\end{equation}
 At $\eta=0$ one has $\|v_\beta\|=1$, the state is pure and $b=0$: this is the PSC,
and the attached object is again a line. For $\eta>0$ it is a subspace of dimension $>1$.

The eigenbases of $\varrho_0$ and $\varrho_1$ are exchanged by $X$, because
$XZX=-Z$ while $XXX=X$; neither of them is the Hadamard basis, which
diagonalizes the average $\bar\varrho$. Consequently the eigenbasis of
$\varrho_c=\bigotimes_r\varrho_{c_r}$ is obtained by acting with $X^{c}$ on the eigenbasis of 
$\varrho_0^{\otimes n}$. The 
codeword-dependent (conditional) typical subspace of $\varrho_c$ is spanned by the eigenvectors
whose eigenvalue profile is the typical one, of dimension
$2^{(h_2(b)+o(1))n}=d_E$. As we just saw, it moves with the codeword, while the ambient typical
subspace of $\bar\varrho^{\otimes n}$ remains fixed. This is the {\bf probabilistic
counterpart of the equivariant field} $x\mapsto E_x\subset V$ of
Section~\ref{sec:enlarge}, and the harmonic tower there is an exact, finite-$n$
algebraic proxy for it. The Holevo information then equals the
dimension quotient (cf. Eq.~\eqref{eq: finite bound}):
\begin{equation}\label{eq: holevo}
 \chi={\sf S}(\bar\varrho)-{\sf S}(\varrho_\beta)=h_2(a)-h_2(b)=\frac1n\,\log_2\frac{D}{d_E}+o(1).
\end{equation}

It remains to see which pair $(a,b)$ the criterion of \cite{AG26} selects, and
this is where \eqref{eq:AG-improved} comes from. The PGM bit error rate of the
MQC is $\frac12\bigl[1-(1-2p)(1-2\eta)^{2}\bigr]$ \cite[Eq.~(34)]{AG26}, and
Theorem~1 of \cite{AG26} applies as long as it stays below $\delta$; since
raising $\eta$ lowers $\chi$ and raises the error rate, the best choice is the
$\eta$ for which the error rate equals $\delta$. Therefore, we equate 
$4\eta(1-\eta)(1-2p)^{2}=2(\delta-p)(1-2p)$, so that \eqref{eq: b-mqc} becomes
\[
 a=\tfrac12-\sqrt{p(1-p)},
 \qquad
 b=\tfrac12\Bigl(1-\sqrt{1-2(\delta-p)(1-2p)}\Bigr),
\]
and \eqref{eq: holevo} turns into the expression minimized in
\eqref{eq:AG-improved}. Note that these are precisely the substitutions used in
Section~\ref{sec:result} to prove the identity \eqref{eq: HS}: what looked
there like a change of variables is the statement that the two constructions
attach subspaces of the same two dimensions. In the same vein, the observation
that the infimum in \eqref{eq:OAI improved} is attained with $b>0$ in the quantum context translates
to the statement that the optimal qubit channel is strictly mixed,
$\eta>0$, which is \cite[Prop.~25]{AG26}.

We have thus shown that both proofs, \cite{OAI26} and \cite{AG26}, rely on packing subspaces of asymptotically the same dimension. The two arguments diverge only in the sense in which the subspaces are packed: exactly, through the positive definiteness of $(t-\Lambda)K$ and a count of traces, versus probabilistically, through typicality and the strong converse for c.-q. channels (which we interpret using the packing visualization).

\begin{remark}
The subspace packing ideas extend beyond the construction above, and both
\cite{OAI26} and \cite{AG26} obtain a strict improvement of MRRW-2, recovering
that bound itself as a boundary case. The routes are different. In
\cite{OAI26} the construction is implemented on constant-weight layers, with
harmonic spaces of the Johnson scheme in place of the Boolean ones and then
using the Bassalygo--Elias inequality. The argument in 
\cite{AG26} relies on c.-q. channels: a masking operation hides a bit
$m\sim\mathrm{Ber}(\alpha)$ and reveals its parity with the input, enlarging the
output to four dimensions. The two proofs are nonetheless parallel, the mask pattern playing the role of the layer and the revealed parity carrying the
Bassalygo--Elias reduction.
\end{remark}

\section{Complements and further remarks}\label{sec: complements}

1. The passage from one vector to one subspace goes beyond a formal generalization. The natural attempt--compressing multiplication by the distance coordinate to the moving tower--does not yield a certificate, and the positivity mechanism of \cite{BN08}, a product of positive-definite functions supplied by the three-term recurrence, has no analog for $k>0$. The proof in \cite{OAI26} finds a different mechanism, a Gram factorization through operating on individual coordinates, which yields a strictly smaller eigenvalue (cf. Sec.~\ref{sec: two eigenvalues}). Higher harmonic degrees at a code point are of course familiar from the semidefinite programming bounds of  \cite{Schrijver05} and \cite{BachocVallentin08}, but there the base point is fixed, and the residual symmetry leads one to consider the Terwilliger algebra; as a result, the program becomes matrix-valued and three-point. What had not been anticipated is that letting the subspace move with the codeword keeps the certificate two-point and scalar\footnote{Of course, it was also not anticipated that the two-point Delsarte program, believed for half a century to be optimized by the MRRW certificate, is not (a computation of the LP optimum in \cite{BJ01} had suggested that MRRW bound {\em is} the LP optimum, but the amount of the improvement could not be captured by the precision available at the time of that work).}. If the account in \cite{OAI26} of how the argument was produced is accurate, this idea originated with the system.

2. We make a brief remark on the classical picture, which relies on radial functions. As a function of $z$, the zonal vector is a Krawtchouk polynomial of the
distance: since $\chi_S(x)\chi_S(z)=\chi_S(xz)$,
\[
 v_{i,x}(z)=D_i^{-1/2}\sum_{|S|=i}\chi_S(xz)=\Kt_i\bigl(d(x,z)\bigr)
\]
(normalized to $\|\Kt\|=1$). The tridiagonal matrix arising in the proof is the radial
version of the adjacency matrix of the Hamming ball $J^{(0)}$ (with a small correction at the highest level). The Perron eigenvector of this reduced matrix yields the Christoffel-Darboux kernel, and the certificate is constructed as a product of two positive definite kernels. The Perron eigenvalue has the same asymptotic behavior, and the resulting bound is exponentially unchanged (see also Remark~\ref{rem: LP}). More details appear in \cite{BN06}.

3.
Which parts of the classical argument do {\em not} generalize in the new proof? Which polynomial yields an improvement of the MRRW bounds? 

For $k>0$ the Perron eigenvector still plays a major role, although we do not know if it can be expressed in closed form. The transition coefficients are no longer conservative (mass-preserving) \eqref{eq: nonstochastic}, so there is no three-term recurrence on the retained family.

The certificate can be written as $(t-s)$ times a sum of $d_E^{2}$ squares:
\[
 F=(t-s)\tr(P_xP_y)=(t-s)\sum_{\alpha,\beta}
 \bigl|\la e^x_\alpha,e^y_\beta\ra\bigr|^{2}.
\]
 For $k\ge1$ no individual summand is Delsarte-feasible:
$\Stab(x)$ acts irreducibly and nontrivially on $E_x$, so no
equivariant orthonormal basis exists and $\la e^x_\alpha,e^y_\beta\ra$ is not a
function of $d_H(x,y)$. Furthermore, $F$ is not
$(t-s)g^{2}$ for a single polynomial $g$ of degree $\le L$, which is why it was missed 
in the variational analysis of \cite{BN08} optimizing over this family (and resulting in
the Christoffel--Darboux kernel and hence the MRRW polynomial). Whether the
variational approach extends to sums of many squares, including the
positive-definiteness constraint, seems worth a further study.
A related, interesting question is to obtain an explicit form of $F$ as a polynomial of one variable. The two indices, $i$ and $k$, suggest a relation to Dunkl's addition theorem for Krawtchouk polynomials \cite{Dunkl76}, with $i$ indexing the Krawtchouk terms and $k$ the Hahn terms of that formula.

4. 
The spectral approach applies more broadly, including the Johnson space \cite{BN06}, and its generalization implemented in \cite{OAI26} follows the same ideas as above, yielding the improvement of the asymptotic bound MRRW-2 for constant-weight codes and for general binary codes. These ideas are further applicable for the real sphere $S^{n-1}(\R)$, where the classical ``radial'' implementation stops at the Kabatiansky--Levenshtein bound, and the discussed extension due to OpenAI yields its strict asymptotic improvement.

The spectral approach in its radial form applies across many homogeneous spaces including projective spaces ${\mathbb P}L^{n-1}$, where $L=\R,  \C$, or $\mathbb H$ \cite{BN06}, the $q$-ary Hamming space for $q\ge 2$, as well as some non-distance-transitive spaces such as the Grassmann manifold \cite{Bachoc06} and the ordered (NRT) Hamming space  \cite{BP09}. One would expect that the ideas in \cite{OAI26} could yield some improvements of the bounds computed for them, although this will be technically more involved and arguably of more limited interest.

\subsection{Where a further improvement would come from}
Since the results are new, it is difficult to say whether the new directions they suggest have been conceptually exhausted, and there are certainly no proofs of any such claims. The remarks in this section should therefore be viewed as speculative possibilities rather than as the outcome of a systematic investigation. 

\subsubsection{}\label{sec: two eigenvalues} Two tridiagonal matrices are related to the moving tower \eqref{eq: tower}, namely, the matrix arising from multiplication by the distance and the matrix $J^{(k)}$. They coincide for $k=0$ but not for larger $k$. The spectral radius of the first one is found by observing that by \eqref{eq: partial}, multiplication by $t(x,\cdot)$ is $\frac1n(\partial+\partial^{*})$ conjugated by $R_x$, so the identities
$\partial^{*}\varphi_ih=\sqrt{Q_i}\,\varphi_{i+1}h$ and
$\partial\varphi_{i+1}h=\sqrt{Q_i}\,\varphi_ih$ of Section~\ref{sec:enlarge}
give
\[
 t(x,\cdot)\,\varphi_{i,x}h
 =\frac{\sqrt{Q_i}}n\,\varphi_{i+1,x}h+\frac{\sqrt{Q_{i-1}}}n\,\varphi_{i-1,x}h.
\]
This gives rise to the tridiagonal matrix $M$ with zero diagonal and $M_{i,i+1}=\sqrt{Q_i}/n$. Computing its asymptotic spectral radius by the argument used in Lemma~\ref{lem:eig}, we obtain $\lmax(M)= 2\sqrt{A-B}+o(1)$. At the same time, the certificate of Section~\ref{sec:enlarge} uses $J^{(k)}$ whose spectral radius is 
\[
 \Gamma_H(a,b)=\frac{2(A-B)}{\sqrt A}=2\sqrt{A-B}\cdot\sqrt{1-B/A} .
\]
The boundary condition in \eqref{eq:OAI improved} relies on $\Gamma_H$, so reducing the gap in the spectral radius is where a further improvement could come from.

\subsubsection{}
The subspaces $\cH_k$ used above appear in each Fourier level exactly once, which makes all the transition coefficients scalars and
the matrix $J^{(k)}$ an ordinary tridiagonal matrix. On the sphere the analogous
uniqueness holds for a wider class of subspaces attached at a point, and
\cite{OAI26} exploit this to obtain a whole hierarchy of bounds, and an analogous construction is still to be found. 
Attaching subspaces from a wider class would require enlarging the
ambient space, and the scalar weights $\sqrt{w_i/Z}$ in $\Psi_x$ would become
matrix blocks\footnote{This should not be confused with the multivariate linear programs for spaces
such as the Grassmannian \cite{Bachoc06} or the ordered Hamming space \cite{BP09}, where the relative position of two points is measured by a vector of values. Here the distance is still a single
number, and only the weights become matrices, but even then the
kernel produced is the scalar function $(t-\Lambda)\tr(P_xP_y)$ of $d_H(x,y)$, so such a construction would not leave the two-point Delsarte program} \footnote{Recent preprint \cite{gay2026honeycombframeworkcodebounds} apparently implements elements of this route within the two-point Delsarte framework and the spectral method of \cite{BN06} (they also give another proof of \eqref{eq: HS}); I have not attempted a close reading.}\cite[Remark~1.6]{OAI26}.

\vspace*{.1in}
The next two remarks concern the approach of \cite{AG26} and its interaction with the LP proof.

\subsubsection{} 
For $\ell\ge2$ let
\[
 \cR_\ell(\delta)=\inf\bigl\{\chi(\sigma_0,\sigma_1)\;:\;
 \sigma_0,\sigma_1\in\mathcal D(\C^{\ell})\ \text{output-symmetric},\
 p_e(\sigma_0,\sigma_1)<\delta\bigr\},
\]
which is a nonincreasing function of $\ell$. By \cite[Remark~5.2]{AG26} the MQC exhausts all
output-symmetric qubit pairs, so $R_{\mathrm{MQC}}$ is the best possible result for $\ell=2$; at $\ell=4$ the masked PSC recovers the second MRRW bound and the masked MQC of \cite[\S5.2]{AG26} improves it. The limit
$\lim_{\ell\to\infty}\cR_\ell(\delta)$ is the natural quantity to pursue; at this point we have no information
about it. Note also that \eqref{eq:AG-improved} uses the PGM error rate only as a {\em sufficient} criterion. The PGM decoder is bitwise optimal only for linear codes \cite[Fact 28]{AG26}, while for general codes \cite{AG26} uses a bound on its error probability, so for a fixed $\ell$ an improvement could
be pursued either through a better decoder or through the choice of the output states.

\subsubsection{}
Suppose a new output-symmetric channel gives a bound below
\eqref{eq:AG-improved}. Will this yield a new Delsarte certificate? The answer is, not immediately. Output symmetry supplies  an equivariant field of conditional typical projectors $\Pi_c$ and the kernel
$K(x,y)=\tr(\Pi_x\Pi_y)$, which is a Hilbert-Schmidt Gram matrix and hence positive definite,  and is a function of $d_H(x,y)$, with $D/d_E=2^{n\chi}$, where $\chi$ is the Holevo information. At the same time, a Delsarte certificate must be nonpositive at distances $\ge\delta n$, whereas $K\ge0$ everywhere. If we were to follow the previous constructions, the sign change has to arise from the factor $t-s$ and positive definiteness of $(t-\Lambda)K$. The channel argument relies on the error probability $p_e<\delta$ and does not provide these ingredients, including the spectral radius of a transition matrix. 
We also note that the quantum channel argument is asymptotic in nature, relying on typical subspaces, while
the Delsarte side of the argument is exact at finite $n$.

This question can also be reversed, namely would a new certificate imply new results within the classical-quantum method? There are many factors that affect the answer, so we will not attempt it here.

\vspace*{.1in}
{\sc Acknowledgment:} Claude Opus 5 was used during the preparation of this note to assist with interpreting the results of \cite{OAI26} and \cite{AG26} and to review an earlier draft.


\end{document}